\newif\ificalp
\icalptrue
\ificalp
\documentclass{llncs}
\usepackage{fullpage}
\renewcommand{\paragraph}[1]{{\vspace{2pt}\noindent\textbf{{#1}}}}
\else
\documentclass{sig-alternate}
\fi

\usepackage[english]{babel}
\usepackage[ansinew]{inputenc}
\usepackage{amsmath}
\usepackage{amssymb}
\usepackage{amsfonts}
\usepackage{color}
\usepackage{url}
\usepackage{graphicx}

\usepackage[lofdepth,lotdepth]{subfig}
\usepackage{tikz}

\usepackage{amsthm}

\makeatletter
\newtheorem*{rep@theorem}{\rep@title}
\newcommand{\newreptheorem}[2]{%
\newenvironment{rep#1}[1]{%
 \def\rep@title{#2 \ref{##1}}%
 \begin{rep@theorem}}%
 {\end{rep@theorem}}}
\makeatother

  \newreptheorem{theo}{Theorem}
  
  \newreptheorem{coro}{Corollary}

	\newtheorem{defi}{Definition}

\ificalp
\else
\newtheorem{example}{Example}
\fi

\newcommand{\PrDzmath}[1]{\Pr(\Dec_{dk'}(T')=0^{\ell'}|T'^{#1}=t^{#1})}
\newcommand{\PrDz}[1]{$\PrDzmath{#1}$}
\newcommand{\PrDxmath}[1]{\Pr(\Dec_{dk'}(T')=x|T'^{#1}=t^{#1})}
\newcommand{\PrDx}[1]{$\PrDxmath{#1}$}

\newcommand{\N}{\mathbb{N}}

\newcommand{\E}{\mathbb{E}}

\newcommand{\Gen}{\mathsf{Gen}}
\newcommand{\Enc}{\mathsf{Enc}}
\newcommand{\Dec}{\mathsf{Dec}}
\newcommand{\Ext}{\mathsf{KeyEx}}
\newcommand{\negl}{\mathsf{negl}}
\newcommand{\poly}{\mathsf{poly}}
\newcommand{\from}{\leftarrow}
\newcommand{\zo}{\{0,1\}}
\newcommand{\zon}{\{0,1\}^n}

\begin{document}
\ificalp
\else
\conferenceinfo{Submitted to PETS}{2015}
\fi
\title{How to Bootstrap Anonymous Communication}
\date{}
\ificalp
\author{Sune K. Jakobsen\thanks{School of Mathematical Sciences and School of Electronic Engineering \&
Computer Science, Queen Mary University of London, Mile End Road, London,
E1 4NS, UK. Email: S.K.Jakobsen@qmul.ac.uk.}
\and Claudio Orlandi\thanks{Department of Computer Science, Aarhus University, IT-Parken, Aabogade 34 8200 Aarhus N, Denmark. Email: orlandi@cs.au.dk. Supported by the Danish National Research Foundation
    and The National Science Foundation of China (grant 61361136003)
    for the Sino-Danish Center for the Theory of Interactive
    Computation and from the Center for Research in Foundations of Electronic Markets (CFEM). }}
\institute{}\vspace{-1cm}
\else
\numberofauthors{2}
\author{
%
%
\alignauthor
Sune K. Jakobsen
       \affaddr{School of Mathematical Sciences and School of Electronic Engineering \&
Computer Science, Queen Mary, University of London}\\
       \affaddr{Mile End Road, London}\\
       \affaddr{E1 4NS, UK}\\
       \email{s.k.jakobsen@qmul.ac.uk}
\alignauthor
Claudio Orlandi
%
}

\fi
  \maketitle 
\begin{abstract}
We ask whether it is possible to anonymously communicate a large amount of data using only public (non-anonymous) communication together with a small anonymous channel. We think this is a central question in the theory of anonymous communication and to the best of our knowledge this is the first formal study in this direction. 

To solve this problem, we introduce the concept of \emph{anonymous steganography}: think of a leaker Lea who wants to leak a large document to Joe the journalist. Using anonymous steganography Lea can embed this document in innocent looking communication on some popular website (such as cat videos on \emph{YouTube} or funny memes on \emph{9GAG}). Then Lea provides Joe with a short key $k$ which, \emph{when applied to the entire website}, recovers the document while hiding the identity of Lea among the large number of users of the website. Our contributions include:
\begin{itemize}
\item Introducing and formally defining \emph{anonymous steganography},
\item A construction showing that anonymous steganography is possible (which uses recent results in circuits obfuscation), 
\item A lower bound on the number of bits which are needed to bootstrap anonymous communication.
\end{itemize}

\end{abstract}


\section{Introduction}

Lea the leaker wants to leak a big document to Joe the 
journalist in an anonymous way\footnote{This naming convention is courtesy of Nadia Heninger.}. Lea has a way of anonymously
communicating a small number of bits to Joe, but the size of 
the document she wants to leak is orders of magnitudes 
greater than the capacity of the anonymous channel between 
them.

In this paper we ask whether it is possible to ``bootstrap'' 
anonymous communication, in the sense that we want to 
construct a ``large'' anonymous channel using only public 
(non-anonymous) communication channels together with a 
``small'' anonymous channel. We find the question to be 
central to the theory of anonymous communication and to the 
best of our knowledge this is the first formal study in this 
direction.

To solve this problem, we introduce a novel cryptographic 
primitive, which we call \emph{anonymous steganography}: the 
goal of (traditional) steganography is to hide that a 
certain communication is taking place, by embedding 
sensitive content in innocent looking traffic (such as 
pictures, videos, or other redundant documents). There is no 
doubt that steganography is a useful tool for Lea the 
leaker: using steganography\footnote{For a background on 
steganographic techniques see e.g.,
~\cite{Fridrich:2009:SDM:1721894}.}
she could send sensitive documents to Joe the journalist in 
such a way that even someone monitoring all internet traffic 
would not be able to notice that this communication is 
taking place.\footnote{Of course this powerful eavesdropper 
could try to apply the decoding procedure of the 
steganographic algorithm to the monitored traffic, but 
combining steganography with cryptography (assume e.g., that 
Lea knows Joe's public key) it is quite easy to make sure 
that the message to be steganographically embedded is 
indistinguishable from random.}

However, steganography alone cannot help Lea if she wants to 
make sure that Joe does not learn her identity, and there is 
a strong demand for solutions which guarantee the anonymity 
of whistleblowers (see e.g., SecureDrop
\footnote{\url{https://freedom.press/securedrop}}).

From a high level point of view, anonymous steganography 
allows Lea to embed some sensitive message into an innocent 
looking document, in such a way that someone looking 
\emph{at the entire website}\footnote{Intuitively, it is crucial for Lea's anonymity that Joe can only decode the entire website \emph{at once}: if Joe had a way of decoding single documents (or portions) he would easily be able to pinpoint which document (and therefore which user) contains the sensitive message.} (or a large portion of it) can 
recover the original message without being able to identify 
which of the documents contains the message.
Unfortunately this is too good to be true, and in Section~
\ref{sec:lower} we prove that it is impossible to construct 
an anonymous steganography scheme unless Lea sends 
a key (of super-logarithmic size) to Joe. The idea is: if 
the scheme is correct at some point the probability that Joe 
outputs $x$ has to increase from polynomially small to $1$. 
Joe can estimate how each message (sent by any of the 
users over the non-anonymous channel) affects this 
probability and concludes that the message which changes 
this probability the most must come from Lea.
Hence, the messages that causes this increase has to be sent 
over an anonymous channel. 

To summarize, in anonymous steganography Lea wants to 
communicate a sensitive (large) message $x$ to Joe. To do 
so, she embeds $x$ in some innocent looking (random) 
document $c$ which she uploads to a popular website (not 
necessarily in an anonymous way). Then Lea produces some 
(short) decoding key $dk$ (which is a function of $c$ and 
all other documents on the website -- or at least a set 
large enough so that her identity is hidden in a large group 
of users, such as ``all videos uploaded last week'') which 
she then communicates to Joe using an anonymous channel. Now Joe is able to recover the original 
message $x$ from the website using the key $dk$, but at the 
same time Joe has no way of telling which 
document contains the message (and therefore which of the 
website user is the leaker). In Section~\ref{sec:definition} 
we formally introduce anonymous steganography and in 
Section~\ref{sec:prot} we 
show how to construct such a scheme.

\paragraph{Related Work.} Practical ways for a leaker to 
communicate anonymously with a journalist is by using e.g., 
the aforementioned SecureDrop, which uses Tor~\cite{Tor}. However, Tor is not secure against end-to-end attacks~\cite{Tor}. Another disadvantage in Tor is that it relies on a network of servers whose only purposes is to make anonymous communication possible. This means that countries can, with some success, block Tor servers~\cite{WL12} and they could make it illegal to host such servers. 


\emph{Message In A Bottle}~\cite{miab13} is a protocol where Lea can encrypt her message under Joe's public key, embed it in an image using steganography and post the image on any blog. Joe will now monitor all blogs to see if someone left a (concealed) message for him. Interestingly~\cite{miab13} shows that this approach is feasible in practice and because Lea can use any blog, it will be costly for e.g. a government to prevent Lea from sending the message to Joe. However, in this protocol Joe learns Lea's identity, which is what we are trying to prevent in our work. 

In \emph{cryptogenography}~\cite{Brody14,ITC14} a group of users cooperate to allow a leaker to publish a message with some reasonable degree of anonymity: here we want that anyone should be able to recover the message from the protocol transcript, but no one (even a computationally unbounded observers) should be able to determine with certainty the identity of the leaker. In other words in cryptogenography we are happy as long as the observer cannot produce evidence which proves with certainty the identity of the leaker (which could be used e.g., in a court case). In~\cite{Brody14} the leaker can publish one bit correctly but no observer can guess the identity of the leaker with probability more than $44\%$. In~\cite{ITC14} instead a different setting is considered, where multiple leakers agree to publish some information while hiding their identity by blending into an arbitrarily large group. The leakers do not need perfect anonymity, but just want to ensure that for each leaker, an observer will never assign a probability greater that $c$ to the event that that person is a leaker. It is shown that for any $\epsilon>0$ and sufficiently large $n$, $n$ leakers can publish $\left(-\frac{\log(1-c)}{c}-\log(e)-\epsilon \right)n$ bits, where $e$ is the base of the natural logarithm. Our work is inspired by the model in~\cite{ITC14}. The main difference is that we assume the adversary has bounded computational power, so we only need one leaker and we get all but negligible anonymity. 

For a survey about anonymity channels, see ~\cite{DD08}. In ~\cite{IKOS06} the authors investigated how an anonymous channel could be used to implement other cryptographic primitives, but not if it could be used to bootstrap a larger anonymous channel. Finally, our positive result is inspired by the clever 
techniques of Hub\'{a}\v{c}ek and Wichs~\cite{HW15} to 
compress communication using obfuscation, and crucially relies on their techniques.

\paragraph{Open problems.} 
Unfortunately our positive result crucially relies on heavy tools such as homomorphic encryption and circuit obfuscation, making it very far from being useful in practice. We leave it as a major open question to construct such schemes using simpler and  more efficient cryptographic tools (perhaps even at the price of relaxing the definition of anonymity).

Other open problems include studying whether the computational complexity for the leaker must depend on the size of the anonymity set if the leaker is given a hash of all the documents, and whether it is possible to construct more efficient protocols if multiple leakers are leaking to Joe at once.

\section{Definitions}
\label{sec:definition}

\paragraph{Notation.}
We write $[x,y]$ with $x<y\in\mathbb{N}$ as a shorthand for $\{x,\ldots, y\}$ and $[x]$ as a shorthand for $[1,x]$. If $v$ is a vector $(v_1,\ldots, v_n)$ then $v_{-i}$ is a vector such that $(v_1,\ldots, v_{i-1}, \bot ,v_{i+1},\ldots v_n)$ and $(v_{-i},v_{i})=v$. A function is \emph{negligible} if it goes to $0$ faster than the inverse of any polynomial. We write $\poly(\cdot)$ and $\negl(\cdot)$ for a generic polynomial and negligible function respectively. $x\from S$ denotes sampling a uniform element $x$ from a set $S$. If $A$ is an algorithm $x\from A$ is the output of $A$ on a uniformly random tape. 
We highlight values $\alpha,\beta,\ldots,$ hardwired in a circuit $C$ using the notation $C[\alpha,\beta,\ldots]$.

\paragraph{Anonymous Steganography.} We define an \emph{anonymous steganography} scheme as a tuple of algorithms  $\pi=(\Gen,\Enc,\Ext,\Dec)$ where\footnote{All algorithms (even when not specified) take as input the security parameter $\lambda$, and the length parameters $\ell,\ell',d,s$.}:

\begin{itemize}
\item $ek \from \Gen(1^\lambda)$ is a randomized algorithm which generates an encoding key.
\item $c \from \Enc_{ek}(x)$ is a randomized algorithm which encodes a secret message $x\in\zo^{\ell'}$ into a (pseudorandom looking) document $c\in\zo^\ell$.\footnote{In our scheme $\ell=\ell'$.}
\item $dk \from \Ext_{ek}(t,i)$ 
takes as input a public vector of documents $t\in(\zo^\ell)^d$, an index $i\in[d]$ such that $t_i=c$,
and extracts a (short) decoding key $dk\in\zo^s$.
\item $x' = \Dec_{dk}(t)$ 
recovers a message $x'$ using the decoding key $dk$ and the public vector of documents $t$ in a deterministic way.
\end{itemize}

\paragraph{How to Use The Scheme.} To use anonymous steganography, Lea generates the encoding key $ek$ using $\Gen$, and then encodes her secret $x$ using $\Enc_{ek}$ to get the ciphertext $c$. She can then upload $c$ to some website.\footnote{For simplicity we will in this example assume that Lea is using a website where everyone is storing documents that are indistinguishable from random. If she is using e.g.~YouTube, she would need to use steganography to get an innocent looking \emph{stegotext}, and Lea and Joe should use the inverse program for extracting messages from stegotext whenever they download documents from the site.} 
She then waits some time, and chooses the set of documents she is hiding among, for example, all files uploaded to this website during that day/week. Lea then downloads all these documents $t$ and finds the index $i$ of her own document in this set. 
Finally she computes $dk\from \Ext_{ek}(t,i)$,
 and uses the small anonymous channel to send $dk$ to Joe together with a pointer to $t$.

\paragraph{Properties of Anonymous Steganography.} We require the following properties: \emph{correctness} (meaning that $x'=x$ with overwhelming probability), \emph{compactness} (meaning that $s<\ell'$) and \emph{anonymity} (meaning that a receiver does not learn any information about $i$). Another natural requirement is \emph{confidentiality} (meaning that one should not be able to learn the message without the decoding key $dk$), but it is easy to see that this follows from \emph{anonymity}. Formal definitions follow:


\begin{defi}[Correctness]\label{def:ascorr} We say an anonymous steganography scheme is \emph{$q$-correct} if for all
$\lambda \in \N, x\in \zo^{\ell'}, i\in [d], t_{-i}\in(\zo^\ell)^{d-1}$, the following holds
$$
\Pr\left[\Dec_{ dk }\left( (t_{-i}, c ) \right)= x\right]\geq q.
$$
where 
$ek \from \Gen(1^\lambda)$, $c\from \Enc_{ek}(x)$, $dk \from \Ext_{ek}( (t_{-i},c) ,i   )$ and the probabilities are taken over all the random coins. We simply say that a scheme is \emph{correct} when $q\geq1-\negl(\lambda)$.
\end{defi}

\newcommand{\Adv}{\mathcal{A}}

\newcommand{\Chal}{\mathcal{C}}
\begin{defi}[Anonymity]\label{def:asanon} Consider the following game between an adversary $\Adv$ and a challenger $\Chal$ :
\begin{enumerate}
\item The adversary $\Adv$ outputs a message $x\in \zo^{\ell'}$, two indices $i_0\neq i_1 \in [d]$, and a vector $t_{-(i_0,i_1)}$;
\item The challenger $\Chal$:
\begin{enumerate} 
\item samples a bit $b\from \zo$;
\item computes $ek\from \Gen(1^\lambda)$, $t_{i_b}\from \Enc_{ek}(x)$ and samples $t_{i_{1-b}}\from \zo^\ell$;
\item computes $dk\from \Ext_{ek}\left((t_{-(i_0,i_1)},(t_{i_0},t_{i_1})),i_b\right)$ 
\item outputs $dk,t$;
\end{enumerate}
\item $\Adv$ outputs a guess bit $g$;
\end{enumerate}
We say $\pi$ satisfies \emph{anonymity} if for all PPT $\Adv$ $\left|\Pr[g=b]-\frac{1}{2}\right|=\negl(\lambda)$. 

\end{defi}

\newcommand{\Obf}{\mathcal{O}}
\newcommand{\iO}{i\mathcal{O}}

\newcommand{\obf}[1]{\bar{#1}}

\newcommand{\Circs}{\mathcal{C}}

\newcommand{\FHGe}{\mathsf{HE.G}}
\newcommand{\FHEn}{\mathsf{HE.E}}
\newcommand{\FHDe}{\mathsf{HE.D}}
\newcommand{\FHEv}{\mathsf{HE.Eval}}
\newcommand{\SE}{\mathsf{SE}}
\newcommand{\SD}{\mathsf{SD}}
\newcommand{\SP}{\zo^lambda}
\newcommand{\SC}{\zo^{ell}}
\newcommand{\SK}{\zo^{ell'}}

\newcommand{\SSBG}{\mathsf{VC.G}}
\newcommand{\SSBH}{\mathsf{VC.C}}
\newcommand{\SSBO}{\mathsf{VC.D}}
\newcommand{\SSBV}{\mathsf{VC.V}}
\newcommand{\lblock}{{\ell_\mathsf{b}}}
\newcommand{\lhash}{{\ell_\mathsf{c}}}
\newcommand{\lopen}{{\ell_\mathsf{d}}}
\newcommand{\hroot}{\gamma}
\newcommand{\hk}{{ck}}

\paragraph{Building Blocks.} We will need the following ingredients in our construction: 1) an \emph{indistinguishability obfuscator}~\cite{DBLP:conf/focs/GargGH0SW13} $\obf{C}\from \Obf(C)$ which takes any polynomial size circuit $C$ and outputs an obfuscated version $\obf{C}$; 2) A \emph{compact} homomorphic encryption scheme $(\FHGe,\FHEn,\FHDe,\FHEv)$; 3) A pseudorandom function $f$; 4) A vector commitment scheme $(\SSBG,\SSBH,\SSBO,\SSBV)$ which allows to commit to a long string $x$ using $\SSBH$, and where it is possible to decomitt to individual bits of $x$ using $\SSBO$. Crucially, the proof of correct decomitting $\pi^j$ for any bit $j$ has size at most polylog in $|x|$. In addition, we need that the vector commitment scheme is \emph{somewhere statistically binding} according to the definition of Hub\'{a}\v{c}ek and Wichs~\cite{HW15}: in a nutshell, this means that when generating a commitment key $\hk$ it is possible to specify a special position $i$ such that a) any commitment generated using the key $\hk$ is statistically binding for the $i$-th bit of $x$ (this property is crucial to be able to verify these commitments inside circuits obfuscated using $\iO$) and that b) $\hk$ computationally hides the index $i$. Such a vector commitment scheme can be constructed from fully-homomorphic encryption~\cite{HW15}. To keep the paper self-contained, all these tools are formally defined in the rest of this section.
\paragraph{Indistinguishability obfuscation.} We use an \emph{indistinguishability obfuscator} like the one proposed in~\cite{DBLP:conf/focs/GargGH0SW13} such that $\obf{C}\from \Obf(C)$ which takes any polynomial size circuit $C$ and outputs an obfuscated version $\obf{C}$ that satisfies the following property.
\begin{defi}[Indistinguishability Obfuscation] \label{def:io} We say $\Obf$ is an \emph{indistinguishability obfuscator} for a circuit class $\Circs$ if for all $C_0,C_1\in\Circs$ such that $\forall x:C_0(x)=C_1(x)$ and $|C_0|=|C_1|$ it holds that:
\begin{enumerate}
\item $\forall C\in \Circs, \forall x\in\zon,  \Obf(C)(x)=C(x)$;
\item $|\Obf(C)|=\poly(\lambda|C|)$  
\item  for all PPT $\Adv$:
$$
\left | \Pr[\Adv(\Obf(C_0)) = 0]-\Pr[\Adv(\Obf(C_0)) = 1]\right| < \negl(\lambda)
$$
\end{enumerate}
\end{defi}

\paragraph{Homomorphic Encryption (HE).}
Let $(\FHGe,\FHEn,\FHDe)$ be an IND-CPA public-key encryption scheme with an additional algorithm $\FHEv$ 
which on input the public key $pk$, $n$ ciphertexts $c_1,\ldots,c_n$ and a circuit $C:\zon\to\zo$ outputs a ciphertext $c^*$, then we say that:
\begin{defi}[Correctness -- HE] \label{def:fhecor} An HE scheme $(\FHGe,\FHEn,\FHDe,\FHEv)$ is \emph{correct for a circuit class $\Circs$} if for all $C\in\Circs$

$$
\FHDe_{sk}(\FHEv_{pk}(C,\FHEn_{pk}(x_1),\ldots,\FHEn_{pk}(x_n)) = C(x_1,\ldots,x_n) 
$$
\end{defi}
\begin{defi}[Compactness -- HE]\label{def:fhecompact} An HE scheme $(\FHGe,\FHEn,\FHDe,\FHEv)$  is called  \emph{compact} if there exist a polynomial $s\in\poly(\lambda)$ such that the output of $\FHEv(C,c_1,\ldots,c_n)$ is at most $s$ bits long (regardless of the size of the circuit $|C|$ or the number of inputs $n$).
\end{defi}

The first candidate homomorphic encryption for all circuits 
was introduced by Gentry~\cite{DBLP:conf/stoc/Gentry09}. Later Brakerski
and Vaikuntanathan~\cite{DBLP:conf/focs/BrakerskiV11} showed that it is possible to build homomorphic encryption based only on the (reasonable) assumption that the learning with error problem (LWE) is computationally hard.

\paragraph{Pseudorandom Functions.} We need a pseudorandom function $f : \zo^{\lambda} \times \zo^{\lambda} \to \zo$. It is well known that the existence of one way function (implied by the existence of homomorphic encryption) implies the existence of PRFs.

\paragraph{Somewhere Statistically Binding (SSB) Vector Commitment Scheme.}
This primitive was introduced by Hub\'{a}\v{c}ek and Wichs~\cite{HW15} under the name \emph{somewhere statistically binding hash}, but we think that the term \emph{vector commitment scheme} is better at communicating the goal of this primitive.

In a nutshell, a Merkle tree (instantiated with a collision resistant hash function) allows to construct a vector commitment: the commitment is the root of the tree, and to decommit a single leaf one can simply send the (logarithmically many) hashes corresponding to the nodes which are necessary to compute the root from the leaf. Unfortunately this only leads to a computationally binding commitment, which leads to a problem when verifying these commitments inside a circuit obfuscated using indistinguishability obfuscation. The point is, $\iO$ 
only ensures that the obfuscation of two circuits are computationally indistinguishable if the two original circuits compute the same function. Therefore computational binding is not enough since there exist (even if they hard to find) other inputs which make the verification procedure to accept. A \emph{somewhere statistically binding commitment} has the additional property that when the commitment key is generated, an index $i$ is specified as well, and the commitment key ``hides'' this index $i$. Now a commitment to $x$ is computationally binding for all leaves $\neq i$ and statistically binding for the leaf $i$. This allows us to (via a series of hybrids) use this commitment inside a circuit obfuscated using $\iO$.

More formally a SSB vector commitment scheme is composed of the following algorithms:
\begin{description}
\item[Key Generation:] The key generation algorithm $\hk\from \SSBG(1^\lambda,L,i)$ takes as input an integer $L\leq 2^{\lambda}$ and index $i\in [L]$ and outputs a public key $\hk$.
\item[Commit:] The commit algorithm $\SSBH_{\hk} :(\zo^{\lblock})^L \to \zo^{\lhash}$ is a deterministic polynomial time algorithm which takes as input a string $x=(x_1,\ldots,x_{L})\in({\zo^\lblock})^L$ and outputs $\SSBH_{\hk}(x)\in\zo^\lhash$.
\item[Decommit:] The decommit algorithm $\pi\from \SSBO_{\hk}(x,j)$ given the commitment key $\hk$, the input $x\in({\zo^\lblock})^L$ and an index $j\in [L]$, creates a proof of correct decommitment $\pi\in\zo^{\lopen}$
\item[Verify:] The verify algorithm $\SSBV_{\hk}(y,j,u,\pi)$ given the key $\hk$ and $y\in\zo^\lhash$ an integer index $j\in[L]$, a value $u\in{\zo^\lblock}$ and a proof $\pi\in\zo^\lopen$, outputs $1$ for accept (that $y=\SSBH_\hk(x)$ and $x_j=u$) or $0$ for reject.
\end{description}

\begin{defi}[Vector Commitment Scheme -- Correctness]\label{def:vccor} A vector commitment scheme is \emph{correct} if for any $L\leq 2^{\lambda}$ and $i,j\in [L]$, any $\hk\from \SSBG(1^\lambda,L,i)$, $x\in({\zo^\lblock})^{L}$, $\pi\from \SSBO_{\hk}(x,j)$ it holds that $\SSBV_{\hk}(\SSBH_\hk(x),j,x_j,\pi)=1$. 
\end{defi}

\begin{defi}[Vector Commitment Scheme -- Index Hiding]\label{def:vcih} We consider the following game between an attacker $\Adv$ and a challenger $\Chal$:
\begin{itemize}
\item The attacker $\Adv(1^\lambda)$ chooses an integer $L$ and two indices $i_0\neq i_1\in[L]$;
\item The challenger $\Chal$ chooses a bit $b\from \zo$ and sets $\hk\from\SSBG(1^\lambda,L,i_b)$.
\item The attacker $\Adv$ gets $\hk$ and outputs a guess bit $g$.
\end{itemize}
We say a vector commitment scheme is \emph{index hiding} if for all PPT $\Adv$ $$\left|\Pr[g=b]-\frac{1}{2}\right|<\negl(\lambda)$$
\end{defi}

\begin{defi}[Vector Commitment Scheme -- Somewhere Statistically Binding] 
\label{def:vcssb} We say $\hk$ is \emph{statistically binding for index $i$} if there are no $y,u\neq u', \pi,\pi'$ such that 
$$
\SSBV_{\hk}(y,i,u,\pi)=\SSBV_{\hk}(y,i,u',\pi')=1
$$
\end{defi}

In \cite{HW15} it is shown how to construct SSB vector commitments using homomorphic encryption.

\section{A Protocol For Anonymous Steganography}\label{sec:prot}

We start with a high-level description of our protocol (in 
steps) before presenting the actual construction and proving 
that it satisfies our notion of anonymity.

\paragraph{First attempt.} Let the encoding key $ek$ be a 
key for a PRF $f$, and let the encoding procedure be simply 
a ``symmetric encryption'' of $x$ using this PRF. 

In this first attempt we let the decoding key $dk$ be the 
obfuscation of a circuit $C[i,ek,\gamma](t)$. The 
circuit contains two hard-wired secrets, the index of Lea's 
document $i\in[d]$ and the key for the PRF $ek$. It also 
contains the hash of the entire set of documents $\gamma 
=H(t)$. On input a database $t$ the circuit checks if $
\gamma=H(t)$ and if this is the case outputs $x$ by 
decrypting $t_i$ with $ek$. 

Clearly this first attempt fails miserably since the size of 
the circuit is now proportional to the size of the entire 
database $t=d\ell$, which is even larger than the size of 
the secret message $|x|=\ell$.

\newcommand{\mux}{\mathsf{mux}}

\paragraph{Second attempt.} To remove the dependency on the 
number of documents $d$, we include in the decoding key an 
encryption $\alpha=\FHEn_{pk}(i)$ of the index $i$ (using the 
homomorphic encryption scheme), and an obfuscation of a (new) 
circuit $C[ek,sk,\gamma](\beta)$, which contains hardwired 
secrets $ek$ and $sk$ (the secret key for the homomorphic 
encryption scheme), as well as a hash $
\gamma=H(\FHEv(\mux[t],\alpha))$, where the circuit $\mux[t]
(i)$ outputs $t_i$. The circuit $C$ now checks that $
\gamma=H(\beta)$ and if this is the case computes $t_i\from
\FHDe_{sk}(\beta)$ using the secret key of the HE scheme, 
then decrypts $t_i$ using $ek$ and outputs the 
secret message $x$. When Joe receives the decoding key $dk$, 
Joe constructs the circuit $\mux[t]$ (using the public $t$) 
and computes $\beta=\FHEv(\mux[t],\alpha)$. To 
learn the secret, he runs the obfuscated circuit on $\beta$.

In other words, we are now exploiting the compactness of the homomorphic encryption scheme to let Joe compute an encryption of the document $c=t_i$ from the public database $t$ and the encryption of $i$. Since Lea the leaker can predict this ciphertext\footnote{The evaluation algorithm $\FHEv$ can always be made deterministic since we do not need circuit privacy.}, she can construct a circuit which only decrypts when this particular ciphertext is provided as input. However, the size of $\beta$ (and therefore $C$) is proportional to $\poly(\lambda)+\ell$, thus we are still far from our goal.\footnote{Note that the decryption key also contains an encryption of $i$ which depends logarithmically on $d$, but we are going to ignore all logarithmic factors.}

\paragraph{Third attempt.} To remove the dependency from the length of the document $\ell$, we construct a circuit which takes as input an encryption of a single bit $j$ instead of the whole ciphertext. However, we also need to make sure that the circuit only decrypts these particular ciphertexts, and does not help Joe in decrypting anything else. Moreover, the circuit must perform this check in an efficient way (meaning, independent of the size of $\ell$), so we cannot simply ``precompute'' these $\ell$ ciphertexts and hardwire them into $C$.

This is where we use the vector commitment: we let the 
decoding key include a (short) commitment key $\hk$. We include in the 
obfuscated circuit a (short) commitment $\gamma=\SSBH_{\hk}(\beta)$ (where 
$\beta=(\beta^1,\ldots,\beta^\ell)$ is a vector of encryptions of 
bits) and we make sure that the circuit only helps Joe in decrypting these 
$\ell$ ciphertexts (and nothing else). In other words, we obfuscate the circuit 
$C[ek,sk,\hk,\gamma](\beta',\pi',j)$ which first checks if $
\SSBV_{\hk}(\gamma,j,\beta',\pi')=1$ and if this is the case it 
outputs the $j$-th bit of $x$ from the $j$-th bit of the ciphertext $t^j_i\from 
\FHDe_{sk}(\beta')$.\footnote{This means that we need to use a symmetric encryption scheme where it is possible to recover a single bit of the plaintext from a single bit of the ciphertext. This can easily be done by encrypting $x$ bit by bit using the PRF.} We have now almost achieved our goal, since the size of the decoding key is $\poly(\lambda\log(d\ell))$.

\paragraph{Final attempt.} We now have to argue that our scheme is secure. Intuitively, while it is true that the index $i$ is only sent in encrypted form, we have a problem since the obfuscated circuit contains the secret key for the homomorphic encryption scheme, and we therefore need a final fix to be able to argue that the adversary does not learn any information about $i$.

The final modification to our construction is to encrypt the index $i$ twice under two independent public keys. From these encryptions Joe computes two independent encryptions of the bit $t^j_i$ which he inputs to the obfuscated circuits together with proofs of decommitment. The circuit now outputs $\bot$ if any of the two decommitment proofs are incorrect, otherwise the circuit computes and outputs $x^j$ from one of the two encryptions (and ignores the second ciphertext). 

\paragraph{Anonymity.} Very informally, we can now prove that Joe cannot distinguish between the decoding keys computed using indices $i_0$ and $i_1$ in the following way: we start with the case where the decoding key contains two encryptions of $i_0$ (this correspond to the game in the definition with $b=0$). Then we define a hybrid game where we change one of the two ciphertext from being an encryption of $i_0$ with an encryption of $i_1$. In particular, since we change the ciphertext which is ignored by the obfuscated circuit, this does not change the output of the circuit at all (and we can argue indistinguishability since the obfuscated circuit does not contain the secret key for this ciphertext). We also replace the random document $c_{i_1}$ with an encryption of $x$ with a new key for the PRF. Finally we change the obfuscated circuit and let it recover the message $x$ from the second ciphertext. Thanks to the SSB property of the commitment scheme it is possible to prove, in a series of hybrids, that the adversary cannot notice this change. To conclude the proof we repeat the hybrids (in inverse order) to reach a game which is identical to the definition of anonymity when $b=1$.

\paragraph{The Actual Construction.} A complete specification of our anonymous steganography scheme follows. 
\begin{description}
\item[Key Generation:] On input the security parameter $\lambda$ the algorithm $\Gen$ samples a random key $ek\in\zo^{\lambda}$ for the PRF and outputs $ek$.
\item[Encoding:]  On input a message $x\in\zo^\ell$ 
and an encoding key $ek$ the algorithm $\Enc$ outputs an encoded message $c\in\zo^\ell$ where for each bit $j\in[\ell]$, $c^j=x^j\oplus f_{ek}(j)$.
\item[Key Extraction:]  On input the encoding key $ek$, 
the database of documents $t$, and index $i$ such that $t_i=c$ the algorithm $\Ext$ outputs a decoding key $dk$ generated as follows:
\begin{enumerate}
\item For all $u\in \zo$ run
 $(pk_u,sk_u)\from \FHGe(1^\lambda)$ and $\alpha_u \from \FHEn_{pk_u}(i)$.

\item For all $j\in[\ell],u\in\zo$ run $\beta^j_{u}= \FHEv_{pk_u}(\mux[t,j],\alpha_{u})$\footnote{Note that we consider $\FHEv$ to be a deterministic algorithm. This can always be achieved by fixing the random tape of $\FHEv$ to some constant value.}  where the circuit $\mux[t,j](i)$ outputs the $j$-th bit of the $i$-th document $t^j_i$;

\item For all $u\in\zo$ run
$\hk_u \from 
\SSBG(1^\lambda,\ell,0)$
 and  $\gamma_{u}\from \SSBH_{\hk_u}(\beta^1_{u},\ldots,\beta^\ell_{u})$.
\item Pick a random bit $\sigma \in \zo$.
\item Define the circuit $C[ek,\sigma,sk_\sigma,\hk_0,\hk_1,\gamma_0,\gamma_1](\beta'_0,\beta'_1,\pi'_0,\pi'_1,j)$ as follows:
\begin{enumerate}
\item if($\forall u\in \zo: \SSBV_{hk_u}(\gamma_u,j,\beta'_u,\pi'_u)$) output $\FHDe_{sk_\sigma}(\beta'_\sigma)\oplus f_{ek}(j)$;
\item else output $\bot$;
\end{enumerate}

\item Compute an obfuscation $\obf{C}\from\Obf(C_\sigma)$ where $C_\sigma$ is a shorthand for the circuit defined before, padded to length equal to $\max(C,C')$ (where the circuit $C'$ is defined in the proof of security).
\item Output $dk=(pk_0,pk_1,\alpha_0,\alpha_1,\hk_0,\hk_1,\obf{C})$
\end{enumerate}
\item[Decoding:] On input a decoding key $dk$ and a database of document $t$ the algorithm $\Dec$ outputs a message $x'$ in the following way: 
\begin{enumerate}
\item Parse $dk=(pk_0,pk_1,\alpha_0,\alpha_1,\hk_0,\hk_1,\obf{C})$;
\item For all $j\in[\ell],u\in\zo$ run $\beta^j_{u}=\FHEv_{pk_u}(\mux[t,j],\alpha_{u})$;
\item For all $u\in\zo$ run $\gamma_{u}\from \SSBH_{\hk_u}(\beta^1_{u},\ldots,\beta^\ell_{u})$.
\item For all $j\in[\ell],u\in\zo$ compute $\pi^j_u \from \SSBO_{\hk_u}((\beta^1_{u},\ldots,\beta^\ell_{u}),j)$;
\item For all $j\in [\ell]$ output $(x')^j\from \obf{C}(\beta^j_0,\beta^j_1,\pi^j_0,\pi^j_1,j)$;
\end{enumerate}
\end{description}

\begin{theorem} \label{thm:as} If 
a) $f$ is PRF 
b) $(\SSBG,\SSBH,\SSBO,\SSBV)$ is a vector commitment scheme satisfying 
Definitions~\ref{def:vccor},~\ref{def:vcih} and \ref{def:vcssb} 
c) $(\FHGe,\FHEn,\FHDe,\FHEv)$ is a homomorphic encryption scheme 
satisfying Definition~\ref{def:fhecor} and~\ref{def:fhecompact} and 
d) $\Obf$ is an obfuscator for all polynomial size circuits 
satisfying Definition~\ref{def:io}
then the anonymous steganography scheme $(\Gen,\Enc,\Ext,\Dec)$ 
satisfies Definitions~\ref{def:ascorr} and~\ref{def:asanon}.
\end{theorem}


\begin{proof}
\paragraph{Correctness (Definition~\ref{def:ascorr}).} Correctness follows from 
inspection of the protocol. In particular, for each bit $j
\in[\ell]$ it holds that $$\obf{C}(\beta^j_0,\beta^j_1,\pi^j_0,
\pi^j_1,j))=C[ek,\sigma,sk_\sigma,\hk_0,\hk_1,\gamma_0,\gamma_1]
(\beta^j_0,\beta^j_1,\pi^j_0,\pi^j_1,j)$$
thanks to Definition~\ref{def:io} (Bullet 1). It is also true 
(thanks to Definition~\ref{def:fhecor}) that $\forall u\in\zo$ the 
ciphertext $\beta^j_u$ is such that $$\FHDe_{sk_u}(\beta^j_{u}) = 
\mux[t,j](\FHDe_{sk_u}(\alpha_u))=\mux[t,j](i)=t^j_i$$ Now, since 
$t^j_i=x^j\oplus f_{ek}(j)$ it follows that the output $z$ of $
\obf{C}\neq \bot$ is either $\bot$ or $x^j$. Finally, the circuit 
only outputs $\bot$ if $\exists u \in \zo$ s.t.~$\SSBV_{hk_u}
(\gamma_u,j,\beta^j_u,\pi^j_u)=0$. But since 
$$\hk_u \from 
\SSBG(1^\lambda,\ell,0), \gamma_{u}\from 
\SSBH_{\hk_u}(\beta^1_{u},\ldots,\beta^\ell_{u}), \pi^j_u \from \SSBO_{\hk_u}((\beta^1_{u},\ldots,\beta^\ell_{u}),j)$$
 then the probability that $\obf{C}$ (and therefore $\Dec$) outputs $\bot$ is $0$ thanks to Definition~\ref{def:vccor}.

\paragraph{Anonymity (Definition~\ref{def:asanon}).} We prove anonymity using a series of hybrid games. We start with a game which is equivalent to the definition when $b=0$ and we end with a game which is equivalent to the definition when $b=1$. We prove at each step that the next hybrid is indistinguishable from the previous. Therefore, at the end we conclude that the adversary cannot distinguish whether $b=0$ or $b=1$.

\paragraph{Hybrid 0.} This is the same as the definition when $b=0$. In particular, here it holds that $(\alpha_0,\alpha_1)\from (\FHEn_{pk_0}(i_0),\FHEn_{pk_1}(i_0))$.

\paragraph{Hybrid 1.} In the first hybrid we replace $\alpha_{1-\sigma}$  with $\alpha_{1-\sigma}\from \FHEn_{pk_{1-\sigma}}(i_1)$. Note that the circuit $C[ek,\sigma,sk_\sigma,\hk_0,\hk_1,\gamma_0,\gamma_1](\cdot)$ does \emph{not} contain the secret key $sk_{1-\sigma}$, therefore any adversary that can distinguish between Hybrid 0 and 1 can be turned into an adversary which breaks the IND-CPA property of the HE scheme.

\paragraph{Hybrid 2.} In the previous hybrids $t_{i_{1}}$ is a 
random string from $\zo^\ell$. In this hybrid we replace $t_{i_{1}}$
with an encryption of $x$ using a new PRF key $ek'$. That is, for 
each bit $j\in[\ell]$ we set $t^j_{i_{1}}=x^j\oplus f_{ek'}(j)$.
Clearly, any adversary that can distinguish between Hybrid 1 and 
Hybrid 2 can be used to break the PRF.

\paragraph{Hybrid 3.($\tau,\rho$).} We now define a series of $2(\ell+1)$ 
hybrids indexed by $\tau\in[0,\ell],\rho\in\zo$. In Hybrid 3.$(\tau,\rho)$ we replace 
the obfuscated circuit with the circuit
$
C'[\tau,ek,ek',\sigma,sk_0,sk_1,\hk_0,\hk_1,\gamma_0,\gamma_1](\beta'_0,\beta'_1,\pi'_0,\pi'_1,j)$  defined as
\begin{enumerate}
\item if($\exists u\in \zo: \SSBV_{hk_u}(\gamma_u,j,\beta'_u,\pi'_u)=0$) output $\bot$
\item else if($j>\tau$) output $\FHDe_{sk_\sigma}(\beta'_\sigma)\oplus f_{ek}(j)$;
\item else output $\FHDe_{sk_{1-\sigma}}(\beta'_{1-\sigma})\oplus f_{ek'}(j)$;
\end{enumerate}

We use $C'_{\tau}$ as a shorthand for a circuit defined as above which is padded to length $\max(C,C')$.

In addition, we also replace the way the keys for the vector commitment schemes are generated. Remember that in the previous hybrids 
$$
\forall u\in\zo \ \ \hk_u\from \SSBG(1^\lambda,\ell,0)
$$ which are  now replaced with 
$$
\forall u\in\zo \ \ \hk_u\from \SSBG(1^\lambda,\ell,\tau+\rho), $$ 

From inspection it is clear that the circuit obfuscated in Hybrid 3.$(0.0)$ computes the same function as the circuit obfuscated in Hybrid 2 (since $j$ is indexed starting from $1$ we always have $j>\tau$ and the branch (3) is never taken), and they are therefore indistinguishable thanks to Definition~\ref{def:io} (Bullet 3).

Next, we argue that Hybrid 3.$(\tau,0)$ is indistinguishable from Hybrid 3.$(\tau,1)$ for all $\tau\in [\ell]$. In those hybrids the obfuscated circuit is exactly the same, and the only difference is in the way the commitment keys $\hk_0,\hk_1$ are generated. In particular, the only difference is the index on which the keys are statistically binding. Therefore, any adversary who can distinguish between 3.$(\tau,0)$ and Hybrid 3.$(\tau,1)$ can be used to break the index hiding property (Definition~\ref{def:vcih}) of the vector commitment scheme.

Finally, we argue that Hybrid 3.$(\tau,1)$ is indistinguishable 
from Hybrid 3.$(\tau+1,0)$. First we note that the commitment keys 
$\hk_0,\hk_1$ are identically distributed in these two hybrids 
i.e., in both hybrids $$\forall u\in \zo \ \ \hk_u\from \SSBG(1^\lambda,\ell,\tau+1)$$

The only difference between the two hybrids is what circuits 
are being obfuscated: in Hybrid 3.$(\tau,1)$ we obfuscate $C'_{\tau}$ and in Hybrid 3.$(\tau+1,0)$ we obfuscate $C'_{\tau+1}$. We now argue that 
these two circuits give the same output on every input, and therefore an adversary that can distinguish between Hybrid 3.$(\tau,1)$ and Hybrid 3.$(\tau+1,0)$ can be used to break the indistinguishability obfuscator.

It follows from inspection that the two circuits behave differently only on inputs of the form $(\beta'_0,\beta'_1,\pi'_0,\pi'_1,\tau+1)$. On input of this form:
\begin{itemize}
\item $C'_{\tau}$ (since $j=\tau+1>\tau$) chooses branch (2) and outputs
$$x^j_0 \from \FHDe_{sk_\sigma}(\beta'_\sigma)\oplus f_{ek}(j)$$
\item $C'_{\tau+1}$ (since $j=\tau+1=\tau+1$) chooses branch (3) and outputs $$x^j_1 \from \FHDe_{sk_{1-\sigma}}(\beta'_{1-\sigma})\oplus f_{ek'}(j)$$
\end{itemize}
Now, the statistically binding property of the vector commitment scheme (Definition~\ref{def:vcssb}) allows us to conclude that there exists only one single pair $(\beta'_0,\beta'_1)$ for which $C'_{\tau}$ and $C_{\tau+1}$ do not output $\bot$ (remember that in both hybrids the commitment keys $\hk_0,\hk_1$ are statistically binding on index $\tau+1$), namely the pair
$$
\forall u \in \zo \ \ \beta^j_u = \FHEv_{pk_u}(\mux[t,\tau+1],\alpha_{u})
$$
which decrypts to the pair $(t^j_{i_0},t^j_{i_{1}})$ (since we changed $\alpha_{1-\sigma}$ in Hybrid 1), which in turns were defined as (since we changed $t^j_{i_{1}}$ in Hybrid 2)
$$
(t^j_{i_0},t^j_{i_{1}})=(x^j\oplus f_{ek}(j),x^j\oplus f_{ek'}(j))
$$
which implies that $x^j_0=x^j_1$ and therefore the two circuits have the exact same input output behavior.

This concludes the technical core of our proof, what is left now is to make few simple changes to go from Hybrid 3.$(\ell,0)$ to the same game as Definition~\ref{def:asanon} when $b=1$.

\paragraph{Hybrid 4.} In this hybrid we replace the obfuscated circuit with $$C[ek',\sigma',sk_{\sigma'},\hk_0,\hk_1,\gamma_0,\gamma_1](\cdot)$$ 
where $\sigma'=1-\sigma$. It is easy to see that the input/output behavior of this circuit is exactly the same as $C'_{\ell}$
(since $\forall j\in[\ell] : j \not> \ell$ the circuit $C'_{\ell}$ always executes branch 3) and therefore an adversary that can distinguish between Hybrid 4 and Hybrid 3.$(\ell,0)$ can be used to break the indistinguishability obfuscator. 

\paragraph{Hybrids 5, 6, 7.} In Hybrid 5 we change the distribution 
of both commitment keys $\hk_0,\hk_1$ to $\SSBG(1^\lambda,\ell,0)$ 
(whereas in Hybrid $4$ they were both sampled as 
$\SSBG(1^\lambda,\ell,\ell+1)$). Indistinguishability follows from 
the index hiding property. In Hybrids 6 we replace $t_{i_{0}}$ with a uniformly random string in $\zo^{\ell}$ (whereas in the previous hybrid it was an encryption of $x$ using the PRF $f$ with key $ek$). Since the obfuscated circuit no longer contains $ek$ we can use an adversary which distinguishes between Hybrids 5 and 6 to break the PRF. In Hybrid 7 we replace $\alpha_{1-\sigma'}$ (which in the previous hybrid is an encryption of $i_{0}$) with an encryption of $i_{1}$. Since the obfuscated circuit no longer contains $sk_{1-\sigma'}=sk_{\sigma}$ we can use an adversary which distinguishes between Hybrids 6 and 7 to break the IND-CPA property of the encryption scheme. Now Hybrid 7 is exactly as the definition of anonymity with $b=1$ with a random bit $\sigma'=1-\sigma$ (which is distributed uniformly at random) and a random encoding key $ek'$. This concludes therefore the proof.

 \end{proof}

Our theorem, together with the results of~\cite{HW15} implies the following.
\begin{corollary} Assuming the existence of homomorphic encryption  and indistinguishability obfuscators for all polynomially sized circuits, there exist an anonymous steganography scheme. 
\end{corollary}
\section{Lower Bound}\label{sec:lower}

In this section we show that any (correct) anonymous 
steganography scheme must have a decoding key of size bigger  
than $O(\log(\lambda))$. Since the decoding key must be sent 
over an anonymous channel, this gives a lower 
bound on the number of bits which are necessary to bootstrap 
anonymous communication.

To show this, we find a strategy for Joe that gives him a 
higher probability of guessing the leaker than if he guessed 
uniformly at random. 

Our lower bound applies to a more general class of anonymous steganography schemes than defined earlier, in particular it also applies to \emph{reactive} schemes where the leaker can post multiple documents to the website, as a function of the documents posted by other users.

We define a \emph{reactive anonymous steganography} scheme as a tuple of algorithms  $\pi=(\Enc,\Ext,\Dec)$ where:

\begin{itemize}
\item $(t_k,state_j) \from \Enc_{ek}(x,t^{k-1},state_{j-1})$ is an algorithm which takes as input a message $x\in\zo^{\ell'}$, a sequence of documents $t^{k-1}$ (which represents the set of documents previously sent) and a state of the leaker, and outputs a new document $t_k\in\zo^\ell$, together with a new state.
\item $dk\from \Ext_{ek}(t^d,state)$ is an algorithm which takes as input a transcript of all documents sent and the current state of the leaker and outputs a decryption key $dk\in \{0,1\}^s$.  
\item $x'= \Dec_{dk}(t^d)$ in an algorithm that given transcript $t^d$ returns a guess $x$ of what the secret is in a deterministic way. 
\end{itemize}

To use a reactive anonymous steganography scheme, the leaker's index $i$ is chosen uniformly at random from $\{1,\dots,n\}$ where $n$ is the number of players. For each $k$ from $1$ to $d$ we generate a document $t_k$. If $k\not\equiv i\mod{n}$ we let $t_k\from \{0,1\}^\ell$. This corresponds to the non-leakers sending a message. When $k\equiv i \mod{n}$ we define $(t_k,state_j)\from \Enc_{ek}(x,t^{k-1},state_{j-1})$, where $t^{k-1}=(t_1,\dots,t_{k-1})$. Then we define $dk\from \Ext_{ek}(t^d,state)$ and $x'= \Dec_{dk}(t^d)$. Here $dk$ is the message that Lea would send over the small anonymous channel.\footnote{Note that a ``standard'' anonymous steganography scheme is also a reactive anonymous scheme.}


The definition of $q$-correctness for reactive schemes is the same as for standard schemes, but our definition of anonymity is weaker because we do not allow the adversary to choose the documents for the honest users. This implies that our lower bound is stronger.

\begin{defi}[Correctness]
A reactive anonymous steganography scheme is \emph{$q$-correct} if for all $\lambda$ and $x\in \zo^{\ell'(\lambda)}$ we have
\[\Pr\left[\Dec_{ dk }\left( t^d\right)= x\right]\geq q. \]
where $t$ and $dk$ is chosen as above and the probability is taken over all the random coins.
\end{defi}

\begin{defi}[Weak Anonymity] Consider the following game between an adversary $A$ and a challenger $C$ 
\begin{enumerate}
\item The adversary $A$ outputs a message $x\in \zo^{\ell'}$;
\item The challenger $C$ samples random $i\in[n]$, and generates $t^d, dk$ as described above
\item The challenger $C$ outputs $t^d, dk$
\item $A$ outputs a guess $g$;
\end{enumerate}
We say that an adversary has advantage $\epsilon(\lambda)$  if 
$ \left| \Pr[ g =i] - \frac{1}{n} \right| \geq \epsilon(\lambda) $.
We say a reactive anonymous steganography scheme provides \emph{anonymity} if, for any adversary, the advantage is negligible.
\end{defi}

In the model we assume that the non-leakers' documents are chosen uniformly at random. This is realistic in the case where we use steganography, so that each $t_k$ is the result of extracting information from a larger file. We could also define a more general model where the distribution of each non-leaker's documents $t_k$ depends on the previous transcript. The proof of our impossibility results works as long as the adversary can sample from $T_k|_{T^{k-1}=t^{k-1},i\not\equiv k \mod{n}}$ in polynomial time. Using this general model, we can also model the more realistic situation where the players do not take turns in sending documents, but at each step only send a document with some small probability. To do this, we just consider ``no document'' to be a possible value of $t_k$.

We could also generalise the model to let the leaker use the anonymous channel at any time, not just after all the documents have been sent. However, in such a model, the anonymous channel transmits more information than just the number of bits send over the channel: the times at which the bits are sent can be used to transmit information~\cite{IW10}. For the number of bits sent to be a fair measure of how much information is transferred over the channel, we should only allow the leaker to use the channel when Joe knows she would use the anonymous channel\footnote{That is, there should be a polynomial time algorithm that given previous transcript $t^k$ and previous messages over the anonymous channel decides if the leaker sends a message over the anonymous channel.}, and the leaker should only be allowed to send messages from a prefix-free code (which might depend on the transcript, but should be computable in polynomial time for Joe). Our impossibility result also works for this more general model, however, to keep the notation simple, we will assume that the anonymous channel is only used at the end.

Finally, we could generalise the model by allowing access to public randomness. However, this does not help the players: as none of the players are controlled by the adversary, the players can generate trusted randomness themselves.

We let $T'=(T'_1,\dots, T'_d)$ denote that random variable where each $T'_i$ is uniformly distributed on $\zo^\ell$. In particular $T'|_{T'^k=t^k}$ is the distribution the transcript would follow if the first $k$ documents are given by $t^k$ and all the players were non-leakers. We let $dk'$ be uniformly distributed on $\zo^s$. Joe can sample from both $T'|_{T'^k=t^k}$ and $dk'$ and he can compute $\Dec$. His strategy to guess the leaker given a transcript $t$ will be to estimate \PrDx{k} for each $k\leq d$. That is, given that the transcript of the first $k$ documents is $t^k$ and all later documents is chosen as if the sender was not a leaker and the anonymous channel just sends random bits, what is the probability that the result is $x$? He can estimate this by sampling: given $t^k$ he randomly generates $t^d$ and $dk$, and then he computes $\Dec$ of this extended transcript. 
 
 Joe will now consider how each player affects these probabilities, given by \PrDx{k}. Intuitively, if these probabilities tends to be higher just after a certain player's documents than just before, he would suspect that this player was leaking. Of course, a leaking player might send some documents that lowers \PrDx{k} to confuse Joe, so we need a way to add up all the changes a players does to \PrDx{k}. The simplest idea would be to compute the additive difference $$\PrDxmath{k}-\PrDxmath{k-1}$$ and add these for each player. However, the following example shows that this strategy does not work in general. 

 \begin{example}
 Consider this protocol for two players, where one of them wants to leak one bit. We have $s=0$, that is $dk$ is the empty sting and will be omitted from the notation. 
 First we define the function $\Dec$. This function looks at the two first documents. If none of these are $0^{\ell}$, it returns the first bit of the third document. Otherwise it defines the \emph{leader} to be the first player who send $0^{\ell}$. Next $\Dec$ looks at the first time the leader sent a document different from $0^{\ell}$. If this number represents a binary number less than $\frac{9}{10}\cdot 2^{\ell}$, then $\Dec$ returns the last bit of the document before, otherwise it outputs the opposite value of that bit. If the leader only sends the document $0^{\ell}$ the output of $\Dec$ is just the last bit sent by the other player.
 
 The leaker's strategy is to become the leader. There is extremely small probability that the non-leaker sends $0^{\ell}$ in his first document, so we will ignore this case. Otherwise the leaker sends $0^{\ell}$ in her first document and becomes the leader. When sending her next document, she looks at the last document from the non-leaker. If it ended in $0$, Joe will think there is $90\%$ chance that $0$ it is output and $10\%$ chance that the output will be $1$, and if it ended in $1$ it is the other way around. If the last bit in the non-leakers document is the bit the leakers wants to leak, she just sends the document $0^{\ell-1}1$. To Joe, this will look like the non-leaker raised the probability of this outcome from $50\%$ to $90\%$ and then the leaker raised it to $100\%$. Thus, Joe will guess that the non-leaker was the leaker. 
 
 If the last bit of the previous document was the opposite of what the leaker wanted to reveal, she will ``reset'' by sending $0^{\ell}$. This brings Joe's estimate that the result will be $1$ back to $50\%$. The leaker will continue ``resetting'' until the non-leaker have sent a document ending in the correct bit more times than he has sent a document ending in the wrong bit. For sufficiently high $d$,  this will happen with high probability, and then the leaker sends $0^{\ell}1$. This ensures that $\Dec(T)$ gives the correct value and that Joe will guess that the non-leaker was the leaker.
 
 If the leaker wants to send many bits, the players can just repeat this protocol. 
 \end{example} 
 Obviously, the above protocol for revealing information is not a good protocol: it should be clear to Joe that the leader is not sending random documents.

 As the additive difference does not work, Joe will instead look at the multiplicative factor $$\frac{\PrDxmath{k}}{\PrDxmath{k-1}}.$$ 
 \begin{defi}
 For a transcript $t$ the \emph{multiplicative factor} $mf_{i,[k_0,k_1]}$ of player $j$ over the time interval $[k_0,k_1]$ is given by
 \[mf_{j,[k_0,k_1]}(t,r)=\prod_{[k_0,k_1]\cap (j+n\N)} \frac{\PrDxmath{k})}{\PrDxmath{k-1}},\]
 We also define 
  \[mf_{-i,[k_0,k_1]}(t,r)=\prod_{[k_0,k_1]\setminus (j+n\N)} \frac{\PrDxmath{k}}{\PrDxmath{k-1}},\]
 \end{defi}

 For fixed $k_0$ and non-leaking player $j$ the sequence $$mf_{j,[k_0,k_0]}(T),mf_{j,[k_0,k_0+1]}(T),\dots$$ is a martingale. Furthermore, if we consider the first $k_1-2$ documents to be fixed and player $1$ sends a document at time $k_1-1$ and player $2$ at time $k_1$, then player $1$'s document can affect the distribution of $$mf_{2,[k_0,k_1]}(T')|_{T'^{k_1-1}=t^{k_1-1}}$$ but no matter what document $t_{k_1-1}$ player $1$ sends,  $$mf_{2,[k_0,k_1]}(T')|_{T'^{k_1-1}=t^{k_1-1}}$$ will have expectation $$mf_{2,[k_0,k_1-1]}(t^{k_1-1}).$$ Similar statements holds for the additive difference, but the advantage of the multiplicative factor is that it is non-negative. This, together with the fact that it is also a martingale, implies that it does not get large with high probability.  
 
 \begin{proposition}
 For $j$ and $k_0,k_1$ we have:
 \[\E_{T'|T^{k_1-1}=t^{k_1-1}}mf_{j,[k_0,k_1]}(T)= mf_{j,[k_0,k_1-1]}(t^{k_1-1})\]
 \end{proposition}
 \begin{proof}
 For $k\not\equiv j\mod{n}$ we have $mf_{j,[k_0,k_1]}(t)= mf_{j,[k_0,k_1-1]}(t^{k_1-1})$ for any $t$ so the statement is trivially true. For $k\equiv j\mod{n}$ it follows from Bayes' Theorem.
  \end{proof}
 
 \begin{proposition}
 For fixed $x$ an random $T$ there is probability at most $\frac{4d}{m_0}$ that there exists $j\neq i $ and $k_0$ such that $mf_{j,[k_0,d]}(T)$ or $mf_{-i,[k_0,d]}(T)$ is at least $\frac{m_0}{2}$.
 \end{proposition}
 \begin{proof}
 For fixed $k_0$, and non-leaker $j$ we have $\E\left( mf_{j,[k_0,d]}(T)\right) =1$. As $$mf_{j,[k_0,d]}(t)\geq 0$$ this implies that $$\Pr(mf_{j,[k_0,d]}(T)\geq \frac{m_0}{2}|T)\leq \frac{2}{m_0}$$ Similarly for $mf_{-i,[k_0,d]}$. We have $$mf_{j,[k_0,d]}(t)=mf_{j,[k_0-1,d]}(t)$$ if player $j$ does not send the $k_0$'th document, so for fixed $t$ there are only $d$ different values (not counting $1$) of $mf_{j,[k_0,d]}(t)$ with $j\neq i$ and $k_0\leq d$. By the union bound, the probability that one of the $mf_{j,[k_0,d]}(t)$'s or one of the $mf_{-i,[k_0,d]}(t)$'s are above $\frac{m_0}{2}$ is at most $\frac{4d}{m_0}$.
  \end{proof}
  
   By sampling $T'^d|_{T'^k=t^k}$ and $dk'$ Joe can estimate \PrDx{k} with a small \emph{additive} error, but when the probability is small, there might still be a large \emph{multiplicative} error. In particular, Joe can only do polynomially many samples, so when \PrDx{k} is less than polynomially small Joe will most likely estimate it to be $0$. This is the reason that anonymous steganography with small anonymous channel works at all: we keep \PrDx{k} exponentially small until we use the anonymous channel. Instead, the idea is to estimate the multiplicative factor starting from some time $k_0$ such that  \PrDx{k} is not too small for any $k\geq k_0$. The following proposition is useful when choosing $k_0$ and choosing how many samples we make.

   \begin{proposition}
 Assume that Joe samples $\frac{3\cdot 2^{s+9} d^4}{\epsilon^2}\log\left(\frac{4d}{\epsilon}\right)$ times to estimate \PrDx{k}.
 
 If $\PrDxmath{k}\geq \frac{\epsilon^2}{2^{s+7} d^2}$, there is probability at least $1-\frac{\epsilon}{2d}$ that his estimate will be in the interval 
 \[[(1-\frac{1}{2d})\PrDxmath{k},(1+\frac{1}{2d})\PrDxmath{k}]\]
 \end{proposition}
 \begin{proof}
 Follows from the multiplicative Chernoff bound.
  \end{proof}
  
 \begin{defi}
 In the following we say that Joe's estimate of \PrDx{k} is \emph{bad} if $\PrDxmath{k}\geq \frac{\epsilon^2}{2^{s+7} d^2}$ but his estimate is not in the interval
  \[[(1-\frac{1}{2d})\PrDxmath{k},(1+\frac{1}{2d})\PrDxmath{k}].\]
 \end{defi}
 
Now we are ready to prove the impossibility result.
 
 \begin{theorem}\label{theo:impos}
 Let be $\epsilon$ a function in $\lambda$ such that $\frac{1}{\epsilon}$ is bounded by a polynomial, and let $\pi$ be a reactive anonymous steganography scheme with $s(\lambda)=O(\log(\lambda))$, $\ell'\geq s+7+ 2\log_2(d)-2\log_2(\epsilon)$ that succeeds with probability at least $q(\lambda)$.
Now there is a probabilistic polynomial time Turing machine $A$ that takes input $t$ and $x$ and outputs the leaker identity with probability 
 \[q(\lambda)+\frac{1-q(\lambda)}{n(\lambda)}-\epsilon(\lambda)\]
 \end{theorem}
 

 \begin{proof}
 Let $\pi$ be a reactive anonymous steganography scheme.  We assume that for random $T'$ and $dk'$ the random variable $\Dec_{dk'}(T')$ is uniformly distributed\footnote{If this is not the case, we can define a reactive anonymous scheme $\widetilde{\pi}$ where this is the case: just let $X'$ be uniformly distributed on $\zo^{\ell'}$, let $\widetilde{\Enc}(x,t^{k},state)=\Enc(x\oplus X',t^k,state)$ and $\widetilde{\Dec_{dk}}(t)=X'\oplus \Dec_{dk}(t)$, where $\oplus$ is bitwise addition modulo $2$. To use $\widetilde{\pi}$ we would need $\ell'$ bits of public randomness to give us $X'$. To get this, we can just increase $\ell$ by $\ell'$ and let $X'$ be the last $\ell'$ bits of the first document.} on $\zo^{\ell'}$ and we will just let Joe send $0^{\ell'}$ in the anonymity game. 
 
Let $m_0=\frac{8d}{\epsilon}$. Consider a random transcript $t$. If for some $k_0$ and some non-leaker $j$ we have $mf_{j,[k_0,d]}\geq \frac{m_0}{2}$ or $mf_{-i,[k_0,d]}\geq \frac{m_0}{2}$ we set $E=1$. 
 
 First Joe will estimate \PrDz{k} for all $k$ using $$\frac{3\cdot 2^{s+9} d^4}{\epsilon^2}\log\left(\frac{4d}{\epsilon}\right)$$ samples for each $k$. Set $E=1$ if at least one of these estimates is bad. In all other cases, $E=0$. By the above propositions and the union bound, $\Pr(E=1)\leq \epsilon(\lambda)$.
 
 Now let $k_0$ be the smallest number such that for all $k\geq k_0$ Joe's estimate of \PrDz{k} is at least $\frac{\epsilon^2}{2^{s+7}d^2}$ . The idea would be to estimate the multiplication factors $mf_{j,[k_0+1,d]}$, but the problem is that \PrDz{k_0} could be large (even $1$) even though \PrDz{k_0-1} is small, so the players might not reveal any information after the $k_0-1$'th document. Thus, Joe needs to include the $k_0-1$'th document in his estimate of the multiplication factors, but his estimate of \PrDz{k_0-1} might be off by a large constant factor. To solve this problem, we define
 \[mf_j=
 \begin{cases}
 mf_{j,[k_0+1,d]}&\text{if }j\not\equiv k_0-1\mod{n} \\
 mf_{j,[k_0+1,d]}\frac{\PrDzmath{k_0}}{(1-\frac{1}{2d})^{-1}\frac{\epsilon^2}{2^{s+7} d^2}}&\text{if }j\equiv k_0-1\mod{n} 
 \end{cases}
 \]
 that is, we pretend that $\PrDzmath{k_0}=(1-\frac{1}{2d})^{-1}\frac{\epsilon^2}{2^{s+7} d^2}$ and then use $mf_{j,[k_0,d]}$. We define $mf_{-i}$ the similar way. Joe's estimate of $\Pr(\Dec(T)=X|T^{k_0-1}=t^{k_0-1})$ less that $\frac{\epsilon^2}{2^{s+7} d^2}$, otherwise $k_0$ would have been lower (here we are using the assumption $h\geq s+7+ 2\log_2(d)-2\log_2(\epsilon)$. Without this, $k_0$ could be $1$). Thus, if this estimate it not bad we must have $$\PrDzmath{k_0-1}\leq(1-\frac{1}{2d})^{-1}\frac{\epsilon^2}{2^{s+7} d^2}$$ So if $E=0$ then $mf_j\leq mf_{j,[k_0,d]}\leq \frac{m_0}{2}$. Similar for $mf_{-i}$.
 
If $E=0$ then $mf_j\leq \frac{m_0}{2}$ for all $j\neq i$ and $mf_{-i}\leq \frac{m_0}{2}$. Furthermore, as all Joe's estimate are good, his estimate of $mf_j$ is off by at most a factor $\left(1-\frac{1}{2d}\right)^{-d}< 2$. Now we define Joe's guess: if exactly one of his estimated $mf_j$'s are above $m_0$ he guesses that this player $j$ is the leaker. Otherwise he chooses his guess uniformly at random from all the players. 
There are two ways \PrDz{k} can increase as $k$ increases\footnote{If we allow the leaker to send anonymous bits before the end of the open communication, this is a third way \PrDz{k} can increase. However, if the times where the anonymous channel is used are predictable by Joe, he can still sample as if the anonymous bits where random. This way, each anonymous bits makes \PrDz{k} increase by at most a factor $2$. If the leaker can only send $s$ anonymous bit in total this only moves a factor $2$ increase in \PrDz{k} from a later point in the proof to here.}: by the leaker sending documents or by a non-leaker sending documents. In the cases where $E=0$ and Joe's estimate of $mf_i$ is less than $m_0$ we know that the contribution from the leaker's documents is a factor less than $2m_0$. As $E=0$ we also know that the total contribution from all the non-leakers is at most a factor $\frac{m_0}{2}$. 
So when only $dk'$ has not been revealed to Joe we have $$\Pr(\Dec_{dk'}(T)=X|T=t)< \frac{\epsilon^2}{2^{s+7} d^2} 2m_0\frac{m_0}{2}=\frac{\epsilon^2}{2^{s+6} d^2} m_0^2=2^{-s}$$ 
As the only randomness left to be revealed\footnote{Here we are using that $\Dec$ is deterministic. However, allowing it to be non-deterministic does not help: we could just increase $\ell$ and let $\Dec$ use the extra bits in each document as randomness instead of using a random tape.} is $dk'$ which is uniformly distributed on a set of size $2^{-s}$, we know that $$\Pr(\Dec_{dk'}(T)=0^{\ell'}|T=t)$$ is a multiple of $2^{-s}$. This implies $$\Pr(\Dec_{dk'}(T)=0^{\ell'}|T=t)=0$$ 
In other words, if $\Dec_{dk}(T)=0$ and $E=0$ then $A$ must output $i$. 
Furthermore, in all other cases where $E=0$ Joe will either guess the leaker correctly (because Joe's estimate of $mf_i$ is sufficiently high) or guess uniformly among all the players. The probability that Joe is correct is now
 
 \[\Pr(g=i)\geq q+\frac{1-q}{n}-\Pr(E=1)\geq q+\frac{1-q}{n}-\epsilon.\] 
  \end{proof}

 Notice that we cannot do better than $q+\frac{1-q}{n}$. The players could use a protocol where with probability $q$ the leaker reveals herself and the information and otherwise no-one reveals any information. This protocol succeeds with probability $q$, and when is does, Joe will guess the leaker. With probability $1-q$ it does not succeed, and Joe has probability $\frac{1}{n}$ of guessing the leaker. In total Joe will guess the leaker with probability $q+\frac{1-q}{n}$. Finally we can conclude that:


  \begin{corollary}\label{coro:thecoro}
  If $\pi$ is a reactive anonymous steganography scheme with $s=O(\log(\lambda))$, $d$ polynomial in $\lambda$ and $\frac{\ell'}{\log(\lambda)}\to \infty$ that ensures weak anonymity, then the probability of correctness $q$ tends to $0$ as $\lambda\to\infty$. 
  \end{corollary}

  \begin{proof}
  Let $\pi$ be as in the assumption and define $$\epsilon=\max(\lambda^{-1},2^{-\frac{s+7+2\log_2(d)-\ell'}{2}})$$ By assumption, $s=O(\log(\lambda))$, $\log(d)=O(\log(\lambda))$, and $\frac{\ell'}{\log(\lambda)}\to\infty$, so $\epsilon\to 0$. The parameters satisfy the assumptions in Theorem \ref{theo:impos} so there is an adversary that can guess the leaker with probability 
  \[q+\frac{1-q}{n}-\epsilon=\frac{1}{n}+\frac{n-1}{n}q-\epsilon\geq \frac{1}{2}+\frac{q-2\epsilon}{2}.\]
  As $\pi$ ensures anonymity, $\frac{q-2\epsilon}{2}$ must be negligible and as $\epsilon\to 0$ we must have $q\to 0$. 
   \end{proof}


\bibliographystyle{alpha}

\bibliography{anonsteg} 
\appendix
\newpage

	 \end{document}